\documentclass[reqno,onecolumn,letterpaper,10pt]{amsart}
\usepackage{amsmath,amssymb,fullpage,calc,enumerate}
\usepackage{graphicx,subfigure,psfrag,multicol,caption}
\usepackage{algorithm,algorithmic,array}
\usepackage{setspace}
\usepackage{verbatim}

\usepackage[bookmarks=true]{hyperref}
\usepackage[all]{hypcap}

\allowdisplaybreaks[1] \numberwithin{equation}{section}
\newtheorem{theorem}{Theorem}[section]

\newtheorem{lemma}[theorem]{Lemma}
\newtheorem{corollary}[theorem]{Corollary}

\newcommand\algorithmicinput{\textbf{Input:}}
\newcommand\INPUT{\item[\algorithmicinput]}

\begin{document}

\title{A linear time algorithm for the $3$-neighbour Travelling Salesman Problem on a Halin graph and extensions}


\author[SFU]{Brad Woods}

\author[SFU]{Abraham Punnen}

\author[SFU]{Tamon Stephen}


\address[SFU]{Department of Mathematics, Simon Fraser University, 250-13450 102 Avenue, Surrey, British Columbia, V3T 0A3, Canada}

\begin{abstract}
The Quadratic Travelling Salesman Problem (QTSP) is to find a least cost Hamiltonian cycle in an edge-weighted graph, where costs are defined for all pairs of edges contained in the Hamiltonian cycle. The problem is shown to be strongly NP-hard on a Halin graph. We also consider a variation of the QTSP, called the $k$-neighbour TSP (TSP($k$)).   Two edges $e$ and $f$, $e\neq f$, are  \emph{$k$-neighbours} on a tour $\tau$ if and only if a shortest path (with respect to the number of edges)   between $e$ and $f$ along $\tau$ and containing both $e$ and $f$, has exactly $k$ edges, for $k\geq 2$.  In (TSP($k$)), a fixed nonzero  cost is considered for a pair of distinct edges in the  cost of a tour $\tau$ only when the edges are  $p$-neighbours on  $\tau$ for $2\leq p\leq k$.  We give a linear time algorithm to solve TSP($k$) on a Halin graph for $k = 3$, extending existing algorithms for the cases $k=1,2$.  Our algorithm can be extended further to solve TSP($k$) in polynomial time on a Halin graph with $n$ nodes when $k=O(\log n)$. The possibility of extending our results to some fully reducible class of graphs is also discussed.  TSP($k$) can be used to model the Permuted Variable Length Markov Model in bioinformatics as well as an optimal routing problem for unmanned aerial vehicles (UAVs).
\end{abstract}

\maketitle

\section{Introduction}

The Travelling Salesman Problem (TSP) is to find a least cost Hamiltonian cycle in an edge weighted graph.  It is one of the most widely studied combinatorial optimization problems and is well-known to be NP-hard.  The TSP model has been used in a wide variety of applications.  For details we refer the reader to the well-known books~\cite{Applegate2011,Cook2012, Gutin2002-2,Lawler1985,  Reinelt1994} as well as the papers~\cite{Balas2006,Ergun2006,Larusic2014,Larusic2012}.

For some applications, more than linear combinations of distances between consecutive nodes are desirable in formulating an objective function.  Consider the problem of determining an optimal routing of an unmanned aerial vehicle (UAV) which has a list of targets at specific locations.  This can be modelled as a TSP which requires a tour that minimizes the distance travelled.  However, such a model neglects to take into account the physical limitations of the vehicle, such as turn radius or momentum.  To illustrate this idea, in Figure~\ref{fig:figure1} we give a Hamiltonian path, in Figure~\ref{fig:figure2} we give the corresponding flight path, and Figure~\ref{fig:figure3} shows a route which is longer but can be travelled at a greater speed and hence reducing the overall travel time.  To model the traversal time, we can introduce penalties for pairs of (not necessarily adjacent) edges to force a smooth curve for its traversal.  In this paper we consider a generalization of the TSP which can be used to model similar situations and contains many variations of the TSP, such as the angular-metric TSP~\cite{Aggarwal2000} and the TSP~\cite{Balas95} as special cases.

\begin{figure}[ht]
	\centering

\begin{minipage}[b]{0.27\linewidth}
\centering
\includegraphics[page=1,clip,trim=260 610 255 85,width=\textwidth]{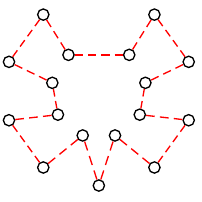}
\caption{Optimal TSP tour with respect to length.}
\label{fig:figure1}
\end{minipage}
\begin{minipage}[b]{0.27\linewidth}
\centering
\includegraphics[page=1,clip,trim=260 450 255 245,width=\textwidth]{figures2.pdf}
\caption{Smoothing of the tour.}
\label{fig:figure2}
\end{minipage}
\begin{minipage}[b]{0.27\linewidth}
\centering
\includegraphics[page=1,clip,trim=260 285 255 405,width=\textwidth]{figures2.pdf}
\caption{A tour which can be travelled more quickly.}
\label{fig:figure3}
\end{minipage}
\end{figure}

Let $G=(V,E)$ be an undirected graph on the node set $V=\{0,1,\ldots ,n-1\}$ with the convention that all indices used hereafter are taken modulo $n$. For each edge $(i,j) \in E$ a nonnegative cost $c_{ij}$ is given.  Let $\mathcal{F}$ be the set of all tours (Hamiltonian cycles) in $G$ and let $\tau=(v_0,v_1,\ldots ,v_{n-1},v_0)\in \mathcal{F}$. 
Note that a tour $\tau$ can be represented either by its sequence (permutation) of the vertices, or by the sequence, or simply the set, of edges it traverses.
In this paper, it is usually more convenient to work with the edge representation.  

The edges $e=(v_i,v_{i+1})$ and $f=(v_j,v_{j+1})$, $e\neq f$, are  \emph{$k$-neighbours} on $\tau$, if and only if a shortest path between $e$ and $f$ on $\tau$ containing these edges has exactly $k$ edges, for $k\geq 2$.  Here the shortest path refers to the path with the least number of edges, rather than the minimum cost path.  Thus $e$ and $f$ are 2-neighbours in $\tau$ if and only if they share a common node in $\tau$.

Let $q(e,f)$ be the cost of the pair $(e,f)$ of edges and $\delta(k,\tau)=\{(e,f) : e,f \in \tau \text{ and } e  \text{ and } f$  are
 $p$-neighbours on  $\tau \text{ for some } 2\leq p\leq k \}$.  Assume that $q(e,f)=q(f,e)$ for every pair of edges $e,f\in E$.
Then the \emph{$k$-neighbour TSP} (TSP($k$)) is defined as in~\cite{Woods2010}

\begin{eqnarray*}
TSP(k): &\textrm{Minimize }   &  \sum_{(e,f)\in \delta(k,\tau)}q(e,f)+\sum_{e\in \tau}c(e)  \\
&\textrm{Subject to } &  \tau \in \mathcal{F}.
\end{eqnarray*}

A closely related problem, the Quadratic TSP (QTSP), is defined as follows:

\begin{eqnarray*}
QTSP: &\textrm{Minimize }   &  \sum_{(e,f)\in \tau\otimes \tau}q(e,f)+\sum_{e\in \tau}c(e)  \\
&\textrm{Subject to } &  \tau \in \mathcal{F}.
\end{eqnarray*}

\noindent where $\tau\otimes \tau = \tau\times \tau \setminus \{(e,e) : e\in \tau \}$. Note:

 \begin{equation*} \tau\otimes \tau =
\begin{cases}
\delta(n/2,\tau) &\mbox{ if $n$ is even}\\
\delta((n+1)/2,\tau) & \mbox{ if $n$ is odd}.
\end{cases}
\end{equation*}

Thus when $k\geq n/2$ (for $n$ even) or $k\geq (n+1)/2$ (for $n$ odd), the $k$-neighbour TSP reduces to the Quadratic TSP~\cite{Woods2016}.  Define TSP($1$) to be the original TSP.  Elsewhere in the literature (e.g. \cite{Fisch2011}, \cite{Fisch2011-2}), the term Quadratic TSP is sometimes used for what we refer to as TSP($2$).  That is, quadratic terms are allowed, but only for pairs of edges that share a node.  

The bottleneck version of TSP($k$) was introduced by Arkin et al.~in~\cite{Arkin:1999}, denoted as the $k$-neighbour maximum scatter TSP.  J\"{a}ger and Molitor~\cite{Jager:2008} encountered TSP(2)
while studying the Permuted Variable Length Markov Model.  Several heuristics are proposed and compared in~\cite{Fischer2014,Jager:2008} as well as a branch and bound algorithm for TSP(2)
in~\cite{Fischer2014}.  A column generation approach to solve TSP($2$) is given in~\cite{Rostami2013}, lower bounding procedures discussed in~\cite{rostami2016lower}, and polyhedral results were
reported by Fischer and Helmberg~\cite{Fisch2011}, Fischer~\cite{Fisch2011-2}, and Fischer and Fischer~\cite{Fischer2015}.  The $k$-neighbour TSP is also related to the $k$-peripatetic salesman
problem~\cite{Della95,Krarup75} and the watchman problem~\cite{Chin86}. Algorithms for maximization and minimization versions of TSP(2) were studied by Stan\v{e}k~\cite{rs} and Oswin et~al.~\cite{Oswin17}.  To the best of our knowledge, no other works in the literature address TSP($k$).

Referring to the UAV example discussed earlier, it is clear that the flight subpaths depend on both the angle and distances between successive nodes.  By precalculating these and assigning costs to $q(e,f)$, for $e,f\in E$, we see that QTSP is a natural model for this problem.  In fact, the flight paths may be affected by edges further downstream.  Thus we can get successively better models by considering TSP($1$), TSP($2$),$\ldots$, TSP($k$) in turn.  In practice we expect diminishing returns to take hold quickly and hence TSP(k) with small values of $k$ are of particular interest.

In this paper we show that QTSP is NP-hard even if the costs are restricted to 0-1 values and the underlying graph is Halin.  In contrast, TSP and TSP(2) on a Halin graph can be solved in $O(n)$
time~\cite{Cornuejols83, Woods2016}.  Interestingly, we show that  TSP($3$) can also be solved on a Halin graph in $O(n)$ time, although as we move from TSP(2) to TSP(3), the problem gets much more
complicated. In fact, our approach can be extended to obtain polynomial time algorithms for TSP($k$) whenever $k = O(\log n)$.  We note that while Halin graphs have treewidth 3, the results on graphs
with bounded treewidth (e.g.~\cite{Bodlaender1988,Courcelle1990}) usually cannot easily be extended to optimization problems with quadratic objective functions.

The paper is organized as follows.  In Section~\ref{HalinSec} we introduce some preliminary results and notations for the problem.  The complexity result for QTSP on Halin graphs is given in Section~\ref{ComplexSec}.  An $O(n)$ algorithm to solve TSP($3$) on Halin graphs is given in Section~\ref{TSPkHalinSec}, which can be extended to obtain an $O(n2^{(k-1)/2})$ algorithm for TSP($k$).  Further extensions of this result to fully reducible classes of graphs are briefly discussed in Section~\ref{ReducibleSec}.

An earlier version of the NP-completeness results presented here were included as part of the M.Sc.\@ thesis of the first author~\cite{Woods2010}.

\section{Notations and definitions}\label{HalinSec}

A Halin graph $H=T\cup C$ is obtained by embedding a tree with no nodes of degree two in the plane and connecting the leaf nodes of $T$ in a cycle $C$ so that the resulting graph remains planar.  Unless otherwise stated, we always assume that a Halin graph or its subgraphs are given in the planar embedded form.  The non-leaf nodes belonging to $T$ are referred to as \textit{tree} or \textit{internal} nodes and the nodes in $C$ are referred to as \textit{cycle} or \textit{outer} nodes of $H$.  A Halin graph with exactly one internal node is called a \textit{wheel}.  If $H$ has at least two internal nodes and $w$ is an internal node of $T$ which is adjacent to exactly one other internal node, then $w$ is adjacent to a set of consecutive nodes of $C$, which we denote by $C(w)$.  Note that $|C(w)|\geq 2$.  The subgraph of $H$ induced by $\{w\} \cup C(w)$ is referred to as a \textit{fan}, and we call $w$ the centre of the fan.  See Figure~\ref{fig:fan}.

\begin{figure}[ht]
	\centering
	\includegraphics[page=1,clip,trim=85 530 90 65]{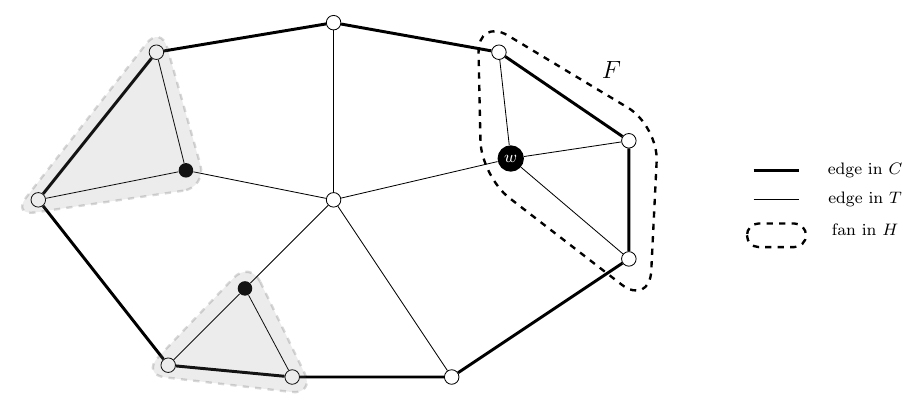}
	\caption{A Halin graph $H$ with 3 fans.  $w$ is the centre of fan $F$.}
	\label{fig:fan}
\end{figure}

\begin{lemma}[Cornuejols et al.~\cite{Cornuejols83}] \label{notwheel2fans}
	Every Halin graph which is not a wheel has at least two fans.
\end{lemma}

Let $G=(V,E)$ be a graph and let $S\subseteq V$ be a connected subgraph of $G$.  Let $\varphi(S)$ be the cutset of $S$, that is, the smallest set of edges whose removal disconnects $S$ from the vertices in $V\setminus S$.  Let $G/S$ be the graph obtained by contracting $S$ into a single node, called a `pseudonode' denoted by $v_S$~\cite{Cornuejols83}.  The edges in $G/S$ are obtained as follows:

\begin{enumerate}
	\item An edge with both ends in $S$ is deleted;
	\item An edge with both ends in $G-S$ remains unchanged;
	\item An edge $(v_1,v_2)$ with $v_1\in G-S$, $v_2\in S$ is replaced by the edge $(v_1, v_s)$.
\end{enumerate}

\begin{lemma}[Cornuejols et al.~\cite{Cornuejols83}] \label{contractHalin}
	If $F$ is a fan in a Halin graph $H$, then $H/F$ is a Halin graph.
\end{lemma}

Note that each time a fan $F$ is contracted using the graph operation $H/F$, the number of non-leaf nodes of the underlying tree is reduced by one.  That is, after at most $\lceil (n-1)/2 \rceil$ fan contractions, a Halin graph will be reduced to a wheel.

Let $w$ be the centre of a fan $F$, and label the outer nodes in $F$ in the order they appear in $C$ as, $u_1,u_2,\ldots ,u_r$ ($r\geq 2$).  Let $(j,k,l)$ be the $3$-edge cutset $\varphi (F)$ which disconnect $F$ from $G$ such that $j$ is adjacent to $u_1$, $k$ is adjacent to $w$ but not adjacent to $u_i$ for any $i$, $1\leq i \leq r$, and $l$ is adjacent to $u_r$ (See Figure~\ref{fig:Halin2}, $r=4$).

Note that every Hamiltonian cycle $\tau$ in $H$ contains exactly two edges of $\{ j,k,l \}$.  The pair of edges chosen gives us a small number of possibilities for traversing $F$ in a tour $\tau$.  For example, if $\tau$ uses $k$ and $l$, it contains the subsequence $w,u_1,u_2,\ldots ,u_r$ (call this a \emph{left-traversal} of $F$), if $\tau$ uses $j$ and $k$ it contains the subsequence $u_1,u_2,\ldots u_r,w$ (call this a \emph{right-traversal} of $F$) and if $\tau$ uses $j$ and $l$, it contains a subsequence of the form $u_1,u_2,\ldots ,u_i,w,u_{i+1},\ldots ,u_r$, for some $i\in \{1,2,\ldots ,r-1\}$ as it must detour through the centre of $F$ (call this a \emph{centre-traversal} of $F$).

\begin{figure}[ht]
	\centering
	\includegraphics[page=1,clip,trim=90 120 95 485]{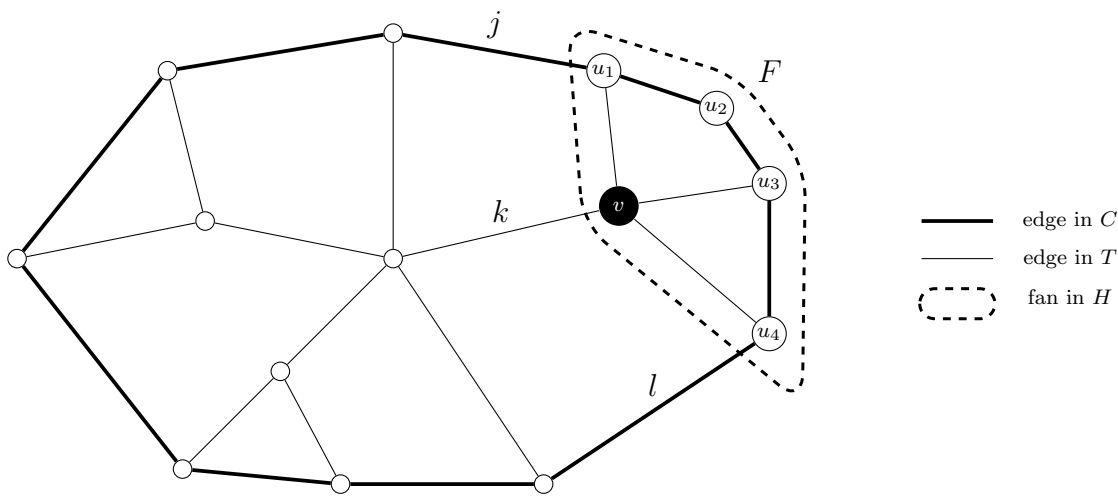}
	\caption{A Halin graph $H$ containing fan $F$. $\{ j,k,l\}$ is a 3-edge cutset which disconnects $F$ from $H$.}
	\label{fig:Halin2}
\end{figure}

\section{Complexity of QTSP on Halin graphs}\label{ComplexSec}

Many optimization problems that are NP-hard on a general graph are solvable in polynomial time on a Halin graph~\cite{Cornuejols83,Lou07,Kabadi98}.  In particular, TSP on a Halin graph is solvable in linear time. Unlike this special case, we show that QTSP is strongly NP-hard on Halin graphs. The decision version of QTSP on a Halin graph, denoted by RQTSP, can be stated as follows:

 ``Given a Halin graph $H$ and a constant $\theta$, does there exist a tour $\tau$ in $H$ such that $  \sum_{e\in \tau} c(e)+ \sum_{e,f \in \tau} q(e,f)  \leq \theta$?''

\begin{theorem}[Woods~\cite{Woods2010}]\label{t1}
    RQTSP is NP-complete even if the values $c(e)\in \{ 0,1\}$ and $q(e,f)\in \{ 0,1\}$ for $e,f\in\nolinebreak H$.
\end{theorem}

\begin{proof} RQTSP is clearly in NP. We now show that the 3-SAT problem can be reduced to RQTSP. The 3-SAT problem can be stated as follows: ``Given a Boolean formula $R$ in Conjunctive Normal Form (CNF) containing a finite number of clauses $C_1,C_2,\ldots, C_h$ on variables $x_1,x_2,\ldots, x_t$ such that each clause contains exactly three literals ($L_1,\ldots,L_{3h}$ where for each $i$, $L_i=x_j$ or $L_i=\neg x_j$ for some $1\leq j\leq t$), does there exists a truth assignment such that $R$ yields a value `true'?''

From a given instance of 3-SAT, we will construct an instance of RQTSP. The basic building block of our construction is a 4-fan gadget obtained as follows. Embed a star on 5 nodes with center $v$ and two specified nodes $\ell$ and $r$ on the plane and add a path $P$ from $\ell$ to $r$ covering each of the pendant nodes so that the resulting graph is planar (see Figure~\ref{fig:gadget}). Call this special graph a 4-fan gadget.

\begin{figure}[ht]
	\centering
	\includegraphics[page=2,clip,trim= 235 580 225 75]{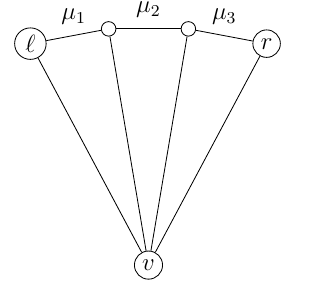}
	\caption{$4$-fan gadget constructed by embedding a star on $5$ nodes in the plane and adding a path.}
	\label{fig:gadget}
\end{figure}

The nodes on path $P$ of this gadget are called {\it outer nodes} and edges on $P$ are called \emph{outer edges}. Let $\mu_1,\mu_2,\mu_3$ be edges with distinct end points in $P$.  Note that any $\ell$-$r$ Hamiltonian path of the gadget must contain all the outer edges except one which is skipped to detour through $v$.  We will refer to an $\ell$-$r$ Hamiltonian path in a $4$-fan gadget as a \emph{center-traversal} as before.

We will construct a Halin graph $H$ using one copy of the gadget for each clause and let $\mu_1,\mu_2,\mu_3$ correspond to literals contained in that clause.  We will assign costs to pairs of edges such that every Hamiltonian cycle with cost 0 must contain a centre-traversal for each clause.  To relate a Hamiltonian cycle to a truth assignment, a centre-traversal which does not contain edge $\mu_i$ corresponds to an assignment of a \emph{true} value to literal $L_i$.

 Now construct $H$ as follows.  For each clause $C_1, \ldots ,C_h$, create a copy of the $4$-fan gadget. The $r, \ell$, and $v$ nodes of  the 4-fan gadget corresponding to the clause $C_i$ are denoted by $r_i,\ell_i$ and $v_i$ respectively. Connect the node $r_i$ to the node $\ell_{i+1}$, $i=1,2,\ldots, h$. Introduce nodes $v_x$ and $v_y$ and the edges $(\ell_1,v_x),(v_x,v_y),(v_y,r_h)$. Also introduce a new node $w$ and connect it to $v_x, v_y$ and $v_i$ for $i=1,2,\ldots, h-1$. The resulting graph is the required Halin graph $H$.  See Figure~\ref{fig:RQTSP}.

\begin{figure}[ht]
	\centering
	\includegraphics[page=2,clip,trim= 175 255 175 270]{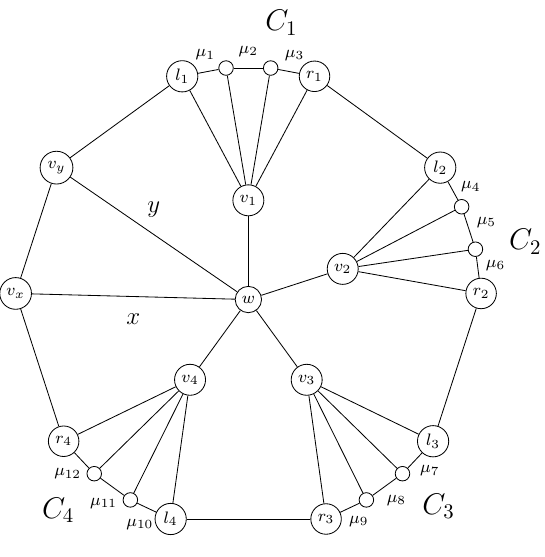}
   \caption{Example of the Halin graph constructed from $F = C_1 \wedge C_2 \wedge C_3 \wedge C_4$.}
   \label{fig:RQTSP}
\end{figure}

  Assign the cost $c(e)=0$ for every edge in $H$.  Let $x=(v_x,w)$ and $y=(v_y,w)$. Note that every tour which contains edges $x$ and $y$ traverses every gadget using a centre-traversal.  For each
gadget: assign costs of $q(e,f) = 1$ for pairs of edges which are neither outer edges nor both adjacent to the same literal edge $\mu_1,\mu_2$ or $\mu_3$, and for all other pairs of edges within the gadget assign cost $0$.  For each variable $x_j$, $j=1, \ldots,n$, and all literals $L_m,L_q$ ($m\neq q$) if $x_j=L_m=\neg L_q$, assign cost $q(\mu'_m,\mu'_q)=1$ where $\mu'_m$ and $\mu'_q$ are edges connecting $\mu_m$ (and $\mu_q$) to the respective $4$-fan gadget centre $v$.  All other paired costs are assumed to be 0.

Suppose $B$ is a valid truth assignment.  Then in each clause there exists at least one true literal.  Consider a tour $\tau$ in $H$ which contains the edges $x$ and $y$ and traverses every gadget such that $\tau$ detours around exactly one literal edge which corresponds to a literal which is \emph{true} in $B$.  Since the truth assignment is valid, such a $\tau$ exists.  Clearly $\tau$ has cost 0, since no costs are incurred by pairs of edges contained in a single gadget, nor are costs incurred of the form $q(\mu'_a, \mu'_b)$ where $L_a = \neg L_b$.  The latter must be true because in any truth assignment, the variable corresponding to $L_a$, say $x_a$, must be either assigned a value of \emph{true} or \emph{false}.  Suppose a cost of 1 is incurred by $q(\mu'_a, \mu'_b)$ and hence $L_a=\neg L_b$.  If $x_a$ is \emph{true}, and $x_a = L_a$, then $L_b$ clearly must be \emph{false}, so $\tau$ cannot detour to miss both $\mu'_a$ and $\mu'_b$.  The same contradiction arises, if $x_a$ is \emph{false}.  Hence a \emph{yes} instance of 3-SAT can be used to construct a \emph{yes} instance for RQTSP with $\theta = 0$.

Now suppose there is a tour which solves RQTSP with $\theta=0$.  Suppose $\tau'$ is such a tour.  Clearly it must use edges $x$ and $y$, and hence must traverse every gadget via a centre-traversal.  Such a detour must skip a literal edge in every gadget, otherwise a cost of 1 is incurred.  Suppose $D = \{ L_1,\ldots,L_s\}$ is the set of literals which are skipped.  $L_i \neq \neg L_j$ for any $i,j$, otherwise a cost of 1 is incurred.  This implies that a truth assignment which results in every literal in $D$ being \emph{true} is a valid truth assignment to the variables $x_1,\ldots ,x_t$.  That is, for each literal edge which is skipped in $\tau'$, assign \emph{true} or \emph{false} to the corresponding variable such that the literal evaluates to \emph{true} (if $L_i=x_j$, set $x_j=true$ and if $L_i=\neg x_j$, set $x_j=false$).  The truth values for any remaining variables can be assigned arbitrarily.  This truth assignment returns \emph{true} for each clause since exactly one literal in each clause is detoured, and evaluates to \emph{true}.  Hence this truth assignment is a valid assignment for 3-SAT. \qed
\end{proof}

\section{Complexity of $k$-neighbour TSP on Halin graphs}\label{ksec}
Let $G$ be a planar embedding of a planar graph and $e$ and $f$ are two distinct edges of $G$. Then $e$ and $f$ are said to be \emph{cofacial}  if there exists a face of $G$ which contains both $e$ and $f$.  This may include the \emph{outer face}.

\begin{theorem}\label{t3}
Let $\tau$ be a tour in the planar embedding $H=T\cup C$ of a Halin graph. Then, any two edges adjacent in $\tau$ must be cofacial.
\end{theorem}

\begin{proof}
Suppose the result is not true. Then, there exists a tour $\tau$ in  $H$ containing two adjacent edges $e=(u,x),f=(x,v)$ such that $e$ and $f$ are not cofacial.  From our previous discussion on fan traversals, we can assume that $x\not\in C$.  Since $e$ and $f$ are not cofacial at $x$, there exists edges $g=(y,x)$ and $h=(x,z)$ in $H$ such that the clockwise ordering of edges incident on $x$ is of the form $f,\ldots,g,\ldots,e,\ldots,h$ (See Figure~\ref{fig:consecHalin}).  Without loss of generality, assume $T$ is rooted at $x$. Then $T$ has at least four subtrees $T_u,T_v,T_y,T_z$ rooted respectively at $u,v,y$ and $z$.  Since $\tau$ is a tour containing the edge $e$, it must contain a path, say $P_1$, through the subtree $T_u$ from $u$ to $u_c\in C$. Note that $u_c$ could be the same as $u$ and in this case the subtree $T_u$ is the isolated node $u$.   Similarly, $\tau$ must contain a path $P_2$ in $T_v$ from $v$ to $v_c\in C$ (See Figure~\ref{fig:consecHalin}).  Note that $P=P_1\cup P_2\cup \{ e,f\}$ is a path in $\tau$.  Deleting the vertex set $V(P)$ of  $P$ and its incident edges  from $H$ yields a disconnected graph. Thus,  $\tau-V(P)$ must be disconnected, a contradiction. \qed
\end{proof}

A preliminary version of this result is given in~\cite{Woods2010}.

\begin{figure}[ht]
	\centering
	\includegraphics[page=3,clip,trim= 180 475 180 60]{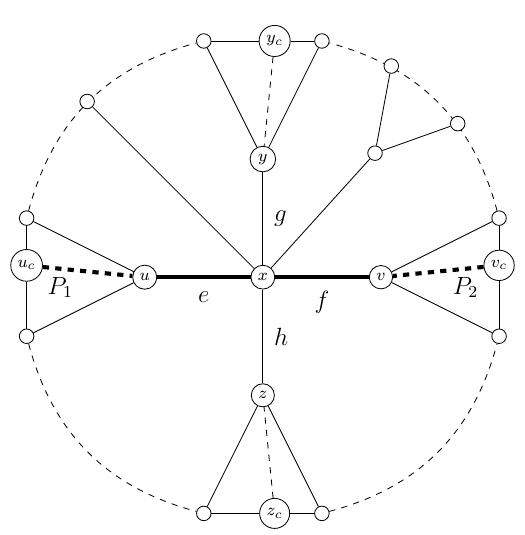}
    \caption{A Halin graph $H$ with non-consecutive edges $e$ and $f$ at node $x$.  $H-P$ where $P=P_1\cup P_2\cup \{ e,f\}$ has two components, hence no $\tau$ contains both edges $e$ and $f$.}
    \label{fig:consecHalin}
\end{figure}

A path in the planar embedding $H=T\cup C$  of a Halin graph is called a \emph{candidate paths} if its consecutive edges are cofacial. A candidate path with $k$-edges is called a \emph{candidate $k$-path}. Note that only candidate paths can be subpaths of a tour but it is possible that there are candidate paths that are not part of a tour.

\begin{corollary} \label{cor:kpaths}
Let $H=T\cup C$ be a  planar embedding of a Halin graph and $e$ be a specified edge of $H$. Then, $H$ has at most $k\cdot2^{k-1}$ candidate $k$-paths containing $e$.
\end{corollary}

\begin{proof}
    Consider extending $e$, if necessary in either directions in $H$, to a candidate $k$-path $P$.  For a given end point of a sub path of $P$, say vertex $u$, by Theorem~\ref{t3}, there are only two possible edges incident on $u$, which may belong to $P$.  Thus, there are at most $2^{k-1}$  candidate $k$-paths when the position of $e$ is fixed.  Since $e$ can take any of the $k$ positions in a candidate $k$-path, there are at most $k\cdot 2^{k-1}$  candidate $k$-paths containing $e$. \qed
\end{proof}

As an immediate consequence of Corollary~\ref{cor:kpaths}, we have an upper bound of $2^{n-2}$ on the number of Hamiltonian cycles in a Halin graph. To see this, let $r$ be the number of Hamiltonian cycles in $H$. From each such cycle, we can generate $n$ distinct Hamiltonian paths (candidate $(n-1)$-paths) by ejecting an edge. Repeating this for all Hamiltonian cycles in $H$,  we get $rn$ candidate $(n-1)$-paths and all these paths are distinct. Hence $r n$ is a lower bound on the number of candidate $(n-1)$-paths in $H$. By Corollary~\ref{cor:kpaths}, $(n-1)2^{n-2}$ is an upper bound on the number of candidate $(n-1)$-paths in $H$. Thus, $r n \leq (n-1)2^{n-2}$ and hence $r\leq \frac{n-1}{n}2^{n-2}\leq 2^{n-2}$.

Noting that there are at most $2(n-1)$ edges in $H$ and that the number of quadratic costs which are relevant is the number of candidate $k$-paths, it follows from Corollary~\ref{cor:kpaths}, that the number of quadratic costs which are relevant is bounded above by $k\cdot 2^{k}\cdot (n-1)=O(n)$ for any fixed $k$ and $O(n^{t+2})$ if $k\leq t\log n$.

Note that any face of $H$ must contain an outer edge.  Moreover, the following Corollary will prove useful.
\begin{corollary}\label{cor:face}
        If $H$ is embedded in the plane such that it is planar and $C$ defines the outer face, for any outer edge $e$ which is contained in the outer face and face $F_e$, every tour which does not contain $e$ must contain all other edges of $F_e$.
\end{corollary}

\subsection{TSP($3$) on Halin graphs}\label{TSPkHalinSec}

As indicated earlier, TSP($1$) is the same as TSP, which is solvable in linear time on Halin graphs~\cite{Cornuejols83}.  TSP($2$) can also be solved in linear time by appropriate modifications of the algorithm of~\cite{Kabadi98} as indicated in \cite{Woods2016}.  However, for  $k\geq 3$, such modifications do not seem a viable option.  We now develop a linear time algorithm to solve TSP($3$).

Let us start with an alternative formulation of TSP(3).  For any subgraph $G$ of $H$, let $P_3(G)$ be the collection of all distinct candidate $3$-paths in $G$.   For each candidate $3$-path $(e,f,g)\in P_3(H)$, define
\begin{equation}\label{eq:stsp}
q(e,f,g) = q(e,g) + \frac{q(e,f) + q(f,g)}{2} + \frac{c(e) + c(f) + c(g)}{3}.
\end{equation}

Now consider the simplified problem:
\begin{eqnarray*}
STSP(3): &\textrm{Minimize }   &  \sum_{(e,f,g) \in P_3(\tau)}q(e,f,g) \\
&\textrm{Subject to } &  \tau \in \mathcal{F}.
\end{eqnarray*}

\begin{theorem}\label{thm:stsp}
    Any optimal solution to the STSP(3) is also optimal solution to TSP(3).
\end{theorem}

\begin{proof}
    For any $\tau \in \mathcal{F}$,
\begin{eqnarray*}
        \sum \limits_{(e,f,g) \in P_3(\tau)} q(e,f,g)
        & = & \sum \limits_{(e,f,g) \in P_3(\tau)} \left(q(e,g) + \frac{q(e,f) + q(f,g)}{2} +  \frac{c(e) + c(f) + c(g)}{3} \right) \\
        & = & \sum \limits_{(e,f) \in \delta(3,\tau)} q(e,f) + \sum \limits_{e \in \delta(\tau)} c(e).
\end{eqnarray*}
Thus, the objective function values of STSP($3$) and TSP($3$) are identical for identical solutions. Since the family of feasible solutions of both these problems are the same, the result follows. \qed
\end{proof}

In view of Theorem~\ref{thm:stsp}, we  restrict our attention to STSP(3).

 For TSP($1$), Cornuejols et al.~\cite{Cornuejols83} identified costs of new edges generated by a fan contraction operation  by solving a linear system of equations. This approach cannot be extended for any $k$-neighbour TSP for $k\geq 3$ as it leads to an over-determined system of equations which may be infeasible.  Instead, we extend the penalty approach used in Phillips et al.~\cite{Kabadi98}.  The idea here is to introduce a node-weighted version of the problem STSP($3$) where we use a penalty function for the nodes of $C$, the value of which depends on the edges chosen to enter and exit the node, along with some other `candidate' edges.  We iteratively contract the fans in $H$, storing the appropriate values of suitable  subpaths as we traverse the fans in a recursive way.  Once we reach a wheel, we can compute an optimal tour for the resulting problem.  Backtracking by recovering appropriate subpaths from contracted fans in sequence,  an optimal solution can be identified.

To formalize the general idea discussed above, let us first discuss the case where $H$ is a Halin graph which is not a wheel.  In this case, $H$ will have at least two fans.

Let $F$ be an arbitrary fan in $H$ with $w$ as the centre. Label the outer nodes of $F$ in the order they appear in $C$, say, $u_1,u_2,\ldots,u_r$ ($r\geq 2$).  Let $\{ j,k,l\}$ be the $3$-edge cutset $\varphi (F)$ which disconnects $F$ from $H$ such that $j$ is adjacent to $u_1$, $k$ is adjacent to $w$ and $l$ is adjacent to $u_r$.  Let $j=(u_1,u_0)$, $k=(w,x)$ and $l=(u_r,u_{r+1})$.  There are exactly two edges not connected to $F$ which are cofacial with $k$ and incident on $x$.   The first edge which follows $k$ in the clockwise orientation of edges incident on $x$ is denoted $\alpha_5$, and the other edge incident on $x$ and cofacial with $k$ is denoted by $\alpha_6$. (See Figure~\ref{fig:fanF}.)  There are exactly two edges not connected to $F$ and incident on $u_0$.  These edges are labelled $\alpha_1,\alpha_2$.  Likewise, there are exactly two edges not connected to $F$ and incident on $u_{r+1}$.  These edges are labelled $\alpha_3,\alpha_4$.  (See Figure~\ref{fig:fanF}.)  Without loss of generality $\alpha_1,\alpha_3$ are in $C$ and $\alpha_2,\alpha_4$ are in $T$.  It is possible that $\alpha_2$ could be the same as $\alpha_5$ and also possible that $\alpha_4$ could be the same as $\alpha_6$.

To complete a fan contraction operation, we consider the 3 types of traversals of $F$. We define a penalty function stored at nodes (pseudonodes) of $C$ which contains attributes of a minimum traversal of $F$ of each type.  For any left- or right-traversal of $F$, there is a single path through $F$ using all cycle edges.  Any tour which includes $j$ and $k$ must pass through one edge of the pair incident on $u_0$ lying outside $F$ together with edges $y_1,\ldots,y_{r-1},t_r,k,\alpha_6$.  Similarly, any tour which includes $k$ and $l$ must pass through one edge of the pair incident on $u_{r+1}$ lying outside $F$ together with edges $y_{r-1},\ldots,y_{1},t_1,k,\alpha_5$.  Any tour which includes $j$ and $l$ must also pass through one edge in each of the pairs of edges incident on $u_0$ (or $u_{r+1}$) lying outside $F$.  That is, every tour $\tau$ containing $j$ and $l$ must contain a path containing one collection of edges from the set $\{ (\alpha_1,j,l,\alpha_3), (\alpha_1,j,l,\alpha_4), (\alpha_2,j,l,\alpha_3),(\alpha_2,j,l,\alpha_4)\}$.  We refer to a centre-traversal of $F$ which bypasses $y_1\in F\cap C$ as a left path, one which bypasses $y_s\in F\cap C$ for some $s\in [2,r-2]$ as a middle path, and one which bypasses $y_{r-1}\in F\cap C$ as a right path.

\begin{figure}[ht]
	\centering
	\includegraphics[page=3,clip,trim= 170 260 170 370]{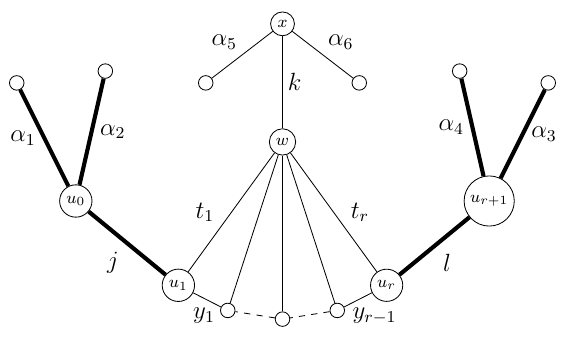}
 	\caption{A fan $F$ with centre $w$.  Every $\tau$ containing $j$ and $k$ contains edges $y_1,\ldots,y_{r-1},t_r$.  Every $\tau$ containing $k$ and $l$ contains edges $y_1,\ldots,y_{r-1},t_1$.  Every $\tau$ containing $j$ and $l$ must contain one of the subpaths from the set $\{ \alpha_1-j-l-\alpha_3, \alpha_1-j-l-\alpha_4, \alpha_2-j-l-\alpha_3,\alpha_2-j-l-\alpha_4\}$.  }
    \label{fig:fanF}
\end{figure}

Let $\mathcal{S}$ be the set of nodes (pseudonodes) in $C$ at some iteration in the contraction process.  In TSP(3), quadratic costs are `absorbed' during each fan contraction operation, depending both on the edges in $F$ and within a distance $2$ from $F$, so care must be taken to retain the proper information.  That is, in order to develop an extension of the penalty approach used in Phillips et al.~\cite{Kabadi98} for TSP(3), we extend the penalty function stored at the nodes (pseudonodes) in $C$ that depends on an additional parameter $\rho$, which specifies the structure of edges around each pseudonode.  Note that due to the recursive property of pseudonodes where a fan contraction operation may `absorb' pseudonodes, the parameter $\rho$ specifies a shape of the structure rather than explicitly stating the edges surrounding a pseudonode.  Further, we define a function $\beta$ which stores the penalty values associated with $\rho$.

Let $A^1$ be the collection of ordered pairs $\{M=(0,0), L=(1,0),R=(0,1),B=(1,1)\}$ and $A^2$ be the collection of ordered pairs $\{ (1,3),(1,4),(1,6),$ $(2,3),(2,4),(2,6),(5,3),(5,4)\}$.  At each node $i\in \mathcal{S}$ we define $\rho_i=(\rho_i^1,\rho_i^2)$ where $\rho_i^1=(\rho_i^{11},\rho_i^{12})\in A^1$ and $\rho_i^2=(\rho_i^{21},\rho_i^{22})\in A^2$.  $\rho_i$ indicates which penalty value (to be defined shortly) stored at pseudonode $i$ is to contribute to the objective function value.  The first component of $\rho_i$ is a binary vector of length 2 which specifies the inner structure of $i$ (edges $y_1,t_1,y_{r-1}$ and $t_r$ prior to any fan contractions such that the first component is $0$ if $y_1$ is selected, and $1$ if $t_1$ is instead, and the second component is $0$ if $y_{r-1}$ is selected, and $1$ if $t_r$ is instead), and the second component of $\rho_i$, the outer structure ($\alpha_1$ to $\alpha_6$ prior to any contraction of adjacent pseudonodes).  Let $\rho$ be the vector containing $\rho_i$ for every $i\in \mathcal{S}$.  Let $\rho_{H/F}$ be the restriction of $\rho$ to the vertices in $H/F$ and augmented by $\rho_{v_F} \in \{(a,b)\}$.  Let $\beta_{i}(\rho_i)$ be the penalty that is incurred if $\rho_i$ occurs at $i$.

For $i\in C$ define:
\begin{equation*}
\mathcal{P}_i(\tau,\rho_i) =
\begin{cases}
	\beta_{i}(\rho_i) & \mbox{ if $\rho_{i}$ is defined,}\\
	0 & \mbox{ otherwise}.
\end{cases}
\end{equation*}

Since the penalties stored at $i\in\mathcal{S}$ depend on edges which are not incident with $i$, and the dependent edges may be `absorbed' into adjacent pseudonodes so that not every $\rho$ is feasible for a given $\tau$.  That is, the inner and outer structures of adjacent pseudonodes in $C$ must agree.  Formally, we say that $\rho$ is feasible for $\tau$ if the following conditions are satisfied for every pseudonode $i$
\begin{enumerate}
\item $j,k\in\tau \iff \rho_i^{22}=6$ and
\item $k,l\in\tau \iff \rho_i^{21}=5$,
\end{enumerate}
and for every pair of consecutive pseudonodes $i,i+1\in C$
\begin{enumerate}
\item $\rho_i^{12} =1 \iff \rho_{i+1}^{21}=2$,
\item $\rho_i^{12} =0 \iff \rho_{i+1}^{21}=1$,
\item $\rho_i^{22}=3 \iff \rho_{i+1}^{11} =0$ and
\item $\rho_i^{22}=4 \iff \rho_{i+1}^{11} =1$.
\end{enumerate}

Let $\mathcal{F'}$ be the set of all feasible $(\tau,\rho)$ pairs.  For an example of a feasible $(\tau,\rho)$ pair, see Figure~\ref{fig:rho}.

\begin{figure}[ht]
	\centering
	\includegraphics[page=4,clip,trim= 145 230 145 460]{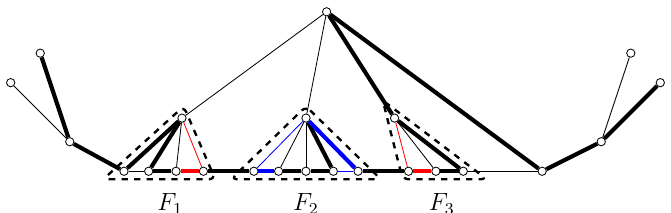}
    \caption{A subgraph of $H$ which becomes a fan after contracting $F_1,F_2$ and $F_3$.  The bold edges depict a centre-traversal of $F$ contained in Hamiltonian cycle $\tau$.  Since $\tau$ contains the edges $t_1$ and $y_{r-1}$ (relative to $F_1$), the inner structure of $v_{F_1}$ must be $(1,0)$, and since $\tau$ contains edges $\alpha_2$ and $\alpha_3$ (relative to $F_1$), the outer structure of $v_{F_1}$ must be $(2,3)$.  That is, a feasible $(\tau,\rho)$ pair must contain $\rho_{v_{F_1}}= ((1,0),(2,3))$.  Considering all of $F_1$, $F_2$ and $F_3$, every feasible $(\tau,\rho)$ pair contains $\rho_{v_{F_1}}=((1,0),(2,3))$, $\rho_{v_{F_2}}=((0,1),(1,3))$ and $\rho_{v_{F_3}}=((0,0),(2,6))$. The edges which correspond to the inner structure of $\rho_{v_{F_2}}$ are coloured in blue and the edges which correspond to the outer structure of $\rho_{v_{F_2}}$ are coloured in red.}
	\label{fig:rho}
\end{figure}

The problem now contains a cost for every triplet of consecutive edges in tour $\tau$ and additionally, a penalty at each outer node $i$.  Consider the modified 3-Neighbour TSP on a Halin graph defined as follows:
\begin{eqnarray*}
MTSP(3): &\textrm{Minimize }   &  z(\tau,\rho) = \sum_{(e,f,g) \in P_3(\tau)}q(e,f,g) + \sum_{i\in C}\mathcal{P}_i(\tau,\rho_i) \\
&\textrm{Subject to } &  (\tau,\rho) \in \mathcal{F'}.
\end{eqnarray*}

The necessary costs to construct MTSP(3) can be obtained as required by applying (\ref{eq:stsp}) or by first embedding $H$ in the plane, and evaluating the candidate $k$-paths from Corollary~\ref{cor:kpaths}.  Note that $\mathcal{P}_v(\tau,\rho_v)$ can be computed in $O(1)$ time by storing penalty 24-tuples containing the $\beta_v$-values described in Table \ref{tab:penalty}, at each cycle node $v$.  Also note that there may be $O(2^n)$ feasible $\rho$ vectors for a given $\tau$, however, we show that the optimal $(\tau,\rho)$ pair can be found in $O(n)$-time.  It is also important to note that the set of pseudonodes is retained for reasons which will become apparent.

\begin{table} \begin{center}
\begin{tabular}{|r|c|l|}  \hline
    & Penalty & Description \\ \hline
     1& $\beta_{v_F}((0,*),(1,6))$ & traversal of $F$ with inner structure $(0,*)$ and outer structure $(1,6)$ \\
     2& $\beta_{v_F}((1,*),(1,6))$ & traversal of $F$ with inner structure $(1,*)$ and outer structure $(1,6)$ \\
     3& $\beta_{v_F}((0,*),(2,6))$ & traversal of $F$ with inner structure $(0,*)$ and outer structure $(2,6)$ \\
     4& $\beta_{v_F}((1,*),(2,6))$ & traversal of $F$ with inner structure $(1,*)$ and outer structure $(2,6)$ \\
     5& $\beta_{v_F}((*,0),(5,3))$ & traversal of $F$ with inner structure $(0,*)$ and outer structure $(5,3)$ \\
     6& $\beta_{v_F}((*,1),(5,3))$ & traversal of $F$ with inner structure $(1,*)$ and outer structure $(5,3)$ \\
     7& $\beta_{v_F}((*,0),(5,4))$ & traversal of $F$ with inner structure $(0,*)$ and outer structure $(5,4)$ \\
     8& $\beta_{v_F}((*,1),(5,4))$ & traversal of $F$ with inner structure $(1,*)$ and outer structure $(5,4)$ \\
     9& $\beta_{v_F}((0,0),(1,3))$ & traversal of $F$ with inner structure $(0,0)$ and outer structure $(1,3)$ \\
    10& $\beta_{v_F}((0,1),(1,3))$ & traversal of $F$ with inner structure $(0,1)$ and outer structure $(1,3)$ \\
    11& $\beta_{v_F}((1,0),(1,3))$ & traversal of $F$ with inner structure $(1,0)$ and outer structure $(1,3)$ \\
    12& $\beta_{v_F}((1,1),(1,3))$ & traversal of $F$ with inner structure $(1,1)$ and outer structure $(1,3)$ \\
    13& $\beta_{v_F}((0,0),(1,4))$ & traversal of $F$ with inner structure $(0,0)$ and outer structure $(1,4)$ \\
    14& $\beta_{v_F}((0,1),(1,4))$ & traversal of $F$ with inner structure $(0,1)$ and outer structure $(1,4)$ \\
    15& $\beta_{v_F}((1,0),(1,4))$ & traversal of $F$ with inner structure $(1,0)$ and outer structure $(1,4)$ \\
    16& $\beta_{v_F}((1,1),(1,4))$ & traversal of $F$ with inner structure $(1,1)$ and outer structure $(1,4)$ \\
    17& $\beta_{v_F}((0,0),(2,3))$ & traversal of $F$ with inner structure $(0,0)$ and outer structure $(2,3)$ \\
    18& $\beta_{v_F}((0,1),(2,3))$ & traversal of $F$ with inner structure $(0,1)$ and outer structure $(2,3)$ \\
    19& $\beta_{v_F}((1,0),(2,3))$ & traversal of $F$ with inner structure $(1,0)$ and outer structure $(2,3)$ \\
    20& $\beta_{v_F}((1,1),(2,3))$ & traversal of $F$ with inner structure $(1,1)$ and outer structure $(2,3)$ \\
    21& $\beta_{v_F}((0,0),(2,4))$ & traversal of $F$ with inner structure $(0,0)$ and outer structure $(2,4)$ \\
    22& $\beta_{v_F}((0,1),(2,4))$ & traversal of $F$ with inner structure $(0,1)$ and outer structure $(2,4)$ \\
    23& $\beta_{v_F}((1,0),(2,4))$ & traversal of $F$ with inner structure $(1,0)$ and outer structure $(2,4)$ \\
    24& $\beta_{v_F}((1,1),(2,4))$ & traversal of $F$ with inner structure $(1,1)$ and outer structure $(2,4)$ \\    \hline
\end{tabular} \caption{Description of penalty 24-tuple stored at pseudonodes in $C$.} \label{tab:penalty}
\end{center}
\end{table}

In the initial graph, and for all $\rho_i$, $i\in C$, set $\beta_{i}(\rho_i)=0$.  For fan $F$ in $H$, the penalties must be updated to store the costs of traversing $F$ when $F$ is contracted to pseudonode $v_F$.  Let $K$ represent the traversal of $F$ which contains only edges in $C$.  That is, $K=j-y_1-\cdots-y_{r-1}-l$.  Then $q(K)=\sum_{e-f-g\in K}q(e,f,g)$ represents the cost incurred by selecting the edges in $K$.  Let $\tau(F)$ and $\rho(F)$ be the restrictions of $\tau$ and $\rho$ to $F$, respectively.

Assign the minimum cost of the right-traversal (which contains $\alpha_s$, $s\in\{1,2\}$ and $\alpha_6$), with inner structure of the first pseudonode $a\in\{L=(1,0),M=(0,0)\}$ to $\beta_{v_F}(a,(s,6))$.  That is, assign to $\beta_{v_F}(a,(s,6))$ the sum of the costs along the traversal, $q(\alpha_s-j-y_1-\cdots-y_{r-1}-t_r-k-\alpha_6)$, together with the minimum feasible set of penalties on the outer nodes contained in $F$, $u_1,u_2\ldots,u_r$.  Note that for the case that $u_1\not\in \mathcal{S}$, it is not possible to have an inner structure $L$ or $B$, and $\beta_{v_F}(L,(s,6))=\beta_{v_F}(B,(s,6))=\infty$.  Otherwise
\begin{align} \label{betaL}
\begin{split}
	\beta_{v_F}(a,(s,6)) &= q(\alpha_s-j-y_1-\cdots-y_{r-1}-t_r-k-\alpha_6) + \min_{\substack{(\tau(F),\rho(F))\in \mathcal{F'}(F):\\\rho_{u_1}^1=a\text{ or } B,\rho_{u_r}^{22}=6}} \left \{  \sum_{i=1}^r \beta_{u_i}(\rho_{u_i}) \right \} \\
    &=  q(\alpha_s,j,y_1) + q(K) -q(y_{r-2},y_{r-1},l) + q(y_{r-2},y_{r-1},t_r) + q(y_{r-1},t_r,k) + q(t_r,k,\alpha_6) \\
	&+ \min_{\substack{(\tau(F),\rho(F))\in \mathcal{F'}(F):\\\rho_{u_1}^1=a \text{ or } B,\rho_{u_r}^{22}=6}} \left \{  \sum_{i=1}^r \beta_{u_i}(\rho_{u_i}) \right \} .
\end{split}
\end{align}
We will explain how the minimum in (\ref{betaL}) can be calculated efficiently later in this paper.

Similarly, assign the minimum cost of the left-traversal (which contains $\alpha_t$, $t\in \{3,4\}$, and $\alpha_5$) with inner structure $a\in\{M=(0,0),R=(0,1)\}$ to $\beta_{v_F}(a,(5,t))$.  In the case that $u_r\not\in \mathcal{S}$, it is not possible to have an inner structure $R$ or $B$, and $\beta_{v_F}(R,(5,t))=\beta_{v_F}(B,(5,t))=\infty$.  Otherwise
\begin{align} \label{betaR}
\begin{split}
	\beta_{v_F}(a,(5,t)) &=  q(\alpha_5,k,t_1) + q(k,t_1,y_1) + q(t_1,y_1,y_2) + q(K) - q(j,y_1,y_2) + q(y_{r-1},l,\alpha_t) \\
	&+ \min_{\substack{(\tau(F),\rho(F))\in \mathcal{F'}(F):\\\rho_{u_r}^1=a\text{ or } B,\rho_{u_1}^{21}=5}} \left \{  \sum_{i=1}^r \beta_{u_i}(\rho_{u_i}) \right \} .
	\end{split}
\end{align}

Let $K(y_i), i\in\{ 1,\ldots,r-1\}$, be the centre-traversal of $F$ which does not contain $y_i$.  Then $q(K(y_i))$ represents the cost incurred by the edges in $K(y_i)$.  That is,
\begin{align*}
	q(K(y_1)) = q(K) + q(j,t_1,t_2) + q(t_1,t_2,y_2) + q(t_2,y_2,y_3) - q(j,y_1,y_2) - q(y_1,y_2,y_3),
\end{align*}
\begin{align*}
	q(K(y_p)) &= q(K) + q(y_{p-2},y_{p-1},t_p) + q(y_{p-1},t_p,t_{p+1}) + q(t_p,t_{p+1},y_{p+1}) + q(t_{p+1},y_{p+1},y_{p+2}) \\&- q(y_{p-2},y_{p-1},y_p) - q(y_{p-1},y_p,y_{p+1}) - q(y_p,y_{p+1}, y_{p+2}),
\end{align*}
for $p\in \{ 2,\ldots,r-2\}$, and
\begin{align*}
	q(K(y_{r-1})) &= q(K) + q(y_{r-3},y_{r-2},t_{r-1}) + q(y_{r-2},t_{r-1},t_r) + q(t_{r-1},t_r,l) \\&- q(y_{r-3},y_{r-2},y_{r-1}) - q(y_{r-2},y_{r-1},l).
\end{align*}

Assign the minimum cost of the centre-traversal which contains $\alpha_s$, $s\in\{1,2\}$, and $\alpha_t$, $t\in\{3,4\}$, which has inner structure $a\in \{L=(1,0),M=(0,0),R=(0,1),B=(1,1)\}$ to $\beta_{v_F}(a,(s,t))$.  In the case that $u_1\not\in \mathcal{S}$ and $a=L$, there is a single path traversing $F$ with inner structure $L$, namely, $j-t_1-t_2-y_2-\cdots-y_{r-1}-l$, so
\begin{align} \label{betaCLnode}
	\beta_{v_F}(L,(s,t)) &= q(\alpha_s,j,t_1) + q(K(y_1)) + q(y_r,l,\alpha_t)
	 + \min_{\substack{(\tau(F),\rho(F))\in \mathcal{F'}(F):\\\rho_{u_r}^{12}=0,\rho_{u_2}^{22}=6}} \left \{  \sum_{i=2}^r \beta_{u_i}(\rho_{u_i}) \right \} ,
\end{align}
and when $u_1\in \mathcal{S}$, we assign the cost of the minimum centre-traversal with inner structure $L$.  Note that this path detours some $y_g$, $g\in \{ 2,\ldots,r-1\}$.
\begin{align} \label{betaCLpnode}
	\beta_{v_F}(L,(s,t)) &= \min_{g\in\{2,\ldots,r-1\}} \left\{q(\alpha_s,j,t_1) + q(K(y_g)) + q(y_r,l,\alpha_t)
	 + \min_{\substack{(\tau(F),\rho(F))\in \mathcal{F'}(F):\\ \rho_{u_1}^{11}=1, \rho_{r}^{12}=0,\\\rho_{u_g}^{21}=5, \rho_{u_{g+1}}^{22}=6}} \left \{  \sum_{i=1}^r \beta_{u_i}(\rho_{u_i}) \right \} \right\}.
\end{align}
Similarly, when $a=R$ and $u_r\not\in \mathcal{S}$, there is a single path traversing $F$ with inner structure $R$, namely, $j-y_1-\cdots-y_{r-2}-r_{r-1}-t_{r}-l$ so
\begin{align} \label{betaCRnode}
	\beta_{v_F}(R,(s,t)) &= q(\alpha_s,j,y_1) + q(K(y_1)) + q(t_r,l,\alpha_t)
	 + \min_{\substack{(\tau(F),\rho(F))\in \mathcal{F'}(F):\\ \rho_{u_1}^{11}=0, \rho_{u_{r-1}}^{22}=6}} \left \{  \sum_{i=1}^{r-1} \beta_{u_i}(\rho_{u_i}) \right \} ,
\end{align}
and when $u_r\not\in \mathcal{S}$, we assign the cost of the minimum centre-traversal with structure $R$.  Note that this path detours some $y_g$, $g\in \{ 1,\ldots,r-2\}$.
\begin{align} \label{betaCRpnode}
	\beta_{v_F}(R,(s,t)) &= \min_{g\in\{1,\ldots,r-2\}} \left\{q(\alpha_s,j,y_1) + q(K(y_g)) + q(y_r,l,\alpha_t)
	 + \min_{\substack{(\tau(F),\rho(F))\in \mathcal{F'}(F):\\\rho_{u_1}^{11}=0, \rho_{r}^{12}=1,\\ \rho_{u_g}^{21}=5, \rho_{u_{g+1}}^{22}=6}} \left \{  \sum_{i=1}^r \beta_{u_i}(\rho_{u_i}) \right \} \right\}.
\end{align}
Now consider $a=M$.  If $u_1\in \mathcal{S},u_r\in \mathcal{S}$ then $\beta_{v_F}(M, (s,t))$ is assigned the cost of the minimum path detouring $y_g$ for $g\in\{ 1,\ldots,r-1\}$.  If $u_1$ ($u_r$) is not a pseudonode then the centre-traversal must contain $y_1$ ($y_{r-1}$) and $1$ ($r-1$) is removed from $g$ from the following equation.
\begin{align} \label{betaCM}
   	\beta_{v_F}(M,(s,t)) &= \min_{g\in\{1,\ldots,r-1\}} \left\{q(\alpha_s,j,y_1) + q(K(y_g)) + q(y_r,l,\alpha_t)
	 + \min_{\substack{(\tau(F),\rho(F))\in \mathcal{F'}(F):\\ \rho_{u_1}^{11}=0, \rho_{u_r}^{12}=0,\\\rho_{u_g}^{21}=5, \rho_{u_{g+1}}^{22}=6}} \left \{  \sum_{i=1}^{r-1} \beta_{u_i}(\rho_{u_i}) \right \} \right\}.
\end{align}
Now consider $a=B$.  This is the same as the case where $a=M$ except for the inner structure.
\begin{align} \label{betaBM}
   	\beta_{v_F}(B,(s,t)) &= \min_{g\in\{1,\ldots,r-1\}} \left\{q(\alpha_s,j,y_1) + q(K(y_g)) + q(y_r,l,\alpha_t)
	 + \min_{\substack{(\tau(F),\rho(F))\in \mathcal{F'}(F):\\\rho_{u_1}^{11}=1, \rho_{u_r}^{12}=1,\\\rho_{u_g}^{21}=5, \rho_{u_{g+1}}^{22}=6}} \left \{  \sum_{i=1}^{r-1} \beta_{u_i}(\rho_{u_i}) \right \} \right\}.
\end{align}

All $\beta_{v_F}$ which have not been assigned are not associated with a feasible $\rho_{v_F}$ and are assigned a value $\infty$.

\begin{theorem}
	Suppose $H=T\cup C$ is a Halin graph which is not a wheel and $F$ is a fan in $H$.  If $(\tau^*,\rho^*)$ is an optimal tour pair in $H$ then there exists a feasible $\rho^*_{H/F}$ in $H/F$ such that $(\tau^*/F,\rho^*_{H/F})$ is optimal in $H/F$ and $z(\tau^*,\rho^*) = z(\tau^*/F,\rho^*_{H/F})$.
\end{theorem}

\begin{proof}
Let $\mathcal{S}$ be the set of pseudonodes in $H$.  Let $v_F$ be the pseudonode which results from the contraction of $F$, and label the vertices and edges of $F$ as in Figure~\ref{fig:fanF}.  Given $(\tau^*,\rho^*)$, we construct $(\tau^*/F,\rho^*_{H/F})$ as follows.  Let $\rho^*_{H/F,i}=\rho^*_i$ $\forall i\not\in F$ and $\rho^*_{H/F,v_F}$ corresponding to the structure of $(\tau^*,\rho^*)$ around $F$ in $H$.  That is, $\tau$ contains an $\alpha_i,\alpha_j$ path through $F$ in $H$, so let
\begin{align*}
	\rho^*_{H/F,v_F} = ((a,b),(c,d))
\end{align*}
where
\begin{align*}
	a =
	\begin{cases}
	0 & \mbox{ if ($u_1\not\in \mathcal{S}$ and $y_1 \in \tau^*$) or ($u_1\in \mathcal{S}$ and $\rho^{*11}_{u_1}=0$)} \\
	1 & \mbox{ if ($u_1\not\in \mathcal{S}$ and $y_1 \not\in \tau^*$) or ($u_1\in \mathcal{S}$ and $\rho^{*11}_{u_1}=1$)},
	\end{cases}
\end{align*}
\begin{align*}
	b =
	\begin{cases}
	0 & \mbox{ if ($u_r\not\in \mathcal{S}$ and $y_{r-1} \in \tau^*$) or ($u_r\in \mathcal{S}$ and $\rho^{*12}_{u_1}=0$)} \\
	1 & \mbox{ if ($u_r\not\in \mathcal{S}$ and $y_{r-1} \not \in \tau^*$) or ($u_r\in \mathcal{S}$ and $\rho^{*12}_{u_1}=1$)},
	\end{cases}
\end{align*}
\begin{align*}
	c =
	\begin{cases}
	 1 & \mbox{ if ($u_1\not\in \mathcal{S}$ and $\alpha_1 \in \tau^*$) or ($u_1\in \mathcal{S}$ and $\rho^{*21}_{u_{0}}=1$)} \\
	 2 & \mbox{ if ($u_{1}\not\in \mathcal{S}$ and $\alpha_2 \in \tau^*$) or ($u_{1}\in \mathcal{S}$ and $\rho^{*21}_{u_{0}}=2$)} \\
	 5 & \mbox{ if ($u_{1}\not\in \mathcal{S}$ and $\alpha_5 \in \tau^*$) or ($u_{1}\in \mathcal{S}$ and $\rho^{*21}_{u_{0}}=5$)},
	\end{cases}
\end{align*}
and
\begin{align*}
	d =
	\begin{cases}
	3 &\mbox{ if ($u_r\not\in \mathcal{S}$ and $t_r \not \in \tau^*$) or ($u_r\in \mathcal{S}$ and $\rho^{*22}_{u_{r}}=3$)} \\
	4 &\mbox{ if ($u_r\not\in \mathcal{S}$ and $t_r \in \tau^*$) or ($u_r\in \mathcal{S}$ and $\rho^{*22}_{u_{r}}=4$)} \\
	6 &\mbox{ if ($u_r\not\in \mathcal{S}$ and $t_r \in \tau^*$) or ($u_r\in \mathcal{S}$ and $\rho^{*22}_{u_{r}}=6$)}.
	\end{cases}
\end{align*}
Then $(\tau^*/F,\rho^*_{H/F,v_F})$ is feasible in $H/F$.

Using equations~(\ref{betaL})-(\ref{betaBM}) and noting that all new triples which contain $v_F$ have $q(e,f,g)=0$, we get
\begin{align*}
	z(\tau^*,\rho^*) &= \sum_{(e,f,g) \in P_3(\tau^*)}q(e,f,g) + \sum_{i\in C}\mathcal{P}_i(\tau^*,\rho^*_i)\\
	&= \sum_{\substack{(e,f,g) \in \\P_3(\tau^*-F)}}q(e,f,g) + \sum_{\substack{(e,f,g) \in \\ P_3(\tau^*) \setminus P_3(\tau^*-F)}}q(e,f,g) + \sum_{i\in C \setminus F} \beta_i(\rho^*_i) + \sum_{i \in F} \beta_i(\rho^*_i) \\
	&= \sum_{\substack{(e,f,g) \in \\P_3(\tau^*/F - v_F)}}q(e,f,g) + \sum_{i\in C \setminus F} \beta_i(\rho^*_i) + \beta_{v_F}(\rho_{v_F}) \\
	&= z(\tau^*/F,\rho^*_{H/F}).
\end{align*}

It remains to show that $(\tau^*/F,\rho^*_{H/F})$ is optimal in $H/F$.  Towards a contradiction, suppose there exists a tour pair $(\tau'/F,\rho'_{H/F})\neq(\tau^*/F,\rho^*_{H/F})$ such that $z(\tau'/F,\rho'_{H/F}) < z(\tau^*/F,\rho^*_{H/F})$.  The calculations of the minimums in equations~(\ref{betaL})-(\ref{betaBM}) imply that $\tau'/F \neq \tau^*/F$.  Using equations~(\ref{betaL})-(\ref{betaBM}), we can expand $F$, extending $(\tau'/F,\rho'_{H/F})$ to $(\tau'',\rho'')$ in $H$ with $z(\tau'/F,\rho'_{H/F})=z(\tau'',\rho'')$.  Then $z(\tau'',\rho'') = z(\tau'/F,\rho'_{H/F}) < z(\tau^*/F,\rho^*_{H/F}) = z(\tau^*,\rho^*)$, contradicting the optimality of $(\tau^*,\rho^*)$.  Hence $(\tau^*/F,\rho^*_{H/F})$ is optimal in $H/F$. \qed
\end{proof}

We now show that $\beta_{v_F}$ can be updated in $O(|F|)$-time, by introducing a structure which allows to optimally chain together the $\beta$-values for consecutive nodes in $F\cap C$.  Refer to the subgraph induced by the nodes $\{ w \} \cup \{ u_a,u_{a+1},\ldots,u_b\}$ as \textit{pseudo-fan} $PF_{a,b}$ for ($1\leq a \leq b \leq r$).  Define the minimum penalty associated with pseudo-fan $PF_{a,b}$ to be $PF_{a,b}(c,d)$ for $a,b\in \{ 1,\ldots ,r\}$ with inner structure $c\in \{ 0,1\}^2$ and outer structure $d=(d_1,d_2)\in  \{ (1,3),(1,4),(2,3),(2,4)\}$.  Note that the inner structure of pseudo-fan $PF_{a,b}$ refers to the edges $y_1,t_1$ and $y_{r-1},t_{r}$ within fans $F_a$ and $F_b$ (which are contracted to pseudonodes $u_a$ and $u_b$), respectively.  The outer structure refers to the edges $\alpha_1,\alpha_2$ and $\alpha_3,\alpha_4$ with respect to fans $F_a$ and $F_b$, if $u_a$ and $u_b$ are considered as pseudonodes.

To compute $PF_{a,b}(c,d)$ we use the following recursion which chains together the minimum $\beta$-values through consecutive nodes in $C$ while maintaining feasibility of $\rho$.

For $c_1,c_2\in \{0,1\}$, $d_1\in \{1,2\}$ and $d_2\in\{ 3,4\}$, and $1 \leq i \leq r-1$
\begin{align} \label{eq:PF(i,i)}
	PF_{1,1}((c_1,c_2),(d_1,d_2)) = \beta_{u_1}((c_1,c_2),(d_1,d_2))
\end{align}
\begin{align}
	PF_{r,r}((c_1,c_2),(d_1,d_2)) = \beta_{u_r}((c_1,c_2),(d_1,d_2))
\end{align}
\begin{align} \label{eq:PF(1,i+1)}
	PF_{1,i+1}((c_1,c_2),(d_1,d_2)) = \min \limits_{\substack{s\in \{0,1\},t\in\{3,4\}}} & \{PF_{1,i}((c_1,s),(d_1,t)) + \beta_{u_{i+1}}((t-3,c_2),(s+1,d_2))\}
\end{align}
and
\begin{align} \label{eq:PF(i-1,r)}
	PF_{i-1,r}((c_1,c_2),(d_1,d_2)) = \min \limits_{\substack{s\in \{0,1\},t\in\{1,2\}}}  \{ \beta_{u_{i-1}}((c_1,t-1),(d_1,s+3)) + PF_{i,r}((s,c_2),(t,d_2))\}
\end{align}

To prove that the recursions defined by (\ref{eq:PF(i,i)})-(\ref{eq:PF(i-1,r)}) are correct, first consider the optimal \\ $PF_{1,i+1}((c_1,c_2),(d_1,d_2))$.  The minimum feasible assignment of penalties for the nodes within pseudo-fan $PF_{1,i+1}$ is simply the minimum among optimal assignments of penalties for the nodes within pseudo-fan $PF_{1,i}$ and the penalty at pseudonode $u_{i+1}$, such that these penalties are feasible.  That is, the outer structure of $PF_{1,i}$ must match the inner structure of $u_{i+1}$ and vice-versa.  By definition, this is precisely $\min \limits_{s\in \{0,1\},t\in\{3,4\}} PF_{1,i}((c_1,s),(d_1,t)) + \beta_{u_{i+1}}((t-3,c_2),(s+1,d_2))$.  An analogous argument holds for $PF_{i-1,r}((c_1,c_2),(d_1,d_2))$.

It is now possible to compute the minimum traversals of $F$ used in equations~(\ref{betaL})-(\ref{betaBM}).  For example, middle paths through $F$ have a cost which minimizes the sum of the penalties on $PF_{1,i-1} + \beta_i$ and $\beta_{i+1} + PF_{i+2,r}$, such that both pairs are feasible.  Note that the minimum $\rho(F)$ has been found while performing the recursion defined by equations (\ref{eq:PF(i,i)})-(\ref{eq:PF(i-1,r)}).  In the worst case, when all outer nodes in $F$ belong to $\mathcal{S}$, we must compute $PF_{1,1},\ldots, PF_{1,r-1}$ and $PF_{r,r},\ldots,PF_{2,r}$, which can be done in $O(|F|)$-time.  By pre-computing these, one can evaluate the minimum traversals of $F$ used in equations~(\ref{betaL})-(\ref{betaBM}), in $O(|F|)$-time and hence one can update $\beta_{v_F}$ in $O(|F|)$-time.

We iteratively perform the fan contraction operation, updating costs and penalties until we are left with a wheel.  The optimal tour in $H$ skirts the cycle $C$ and detours exactly once through centre $w$, skipping exactly one edge of $C$.  Orient the cycle $C$ in the clockwise direction.  $\tau$ contains all edges in $C$ except for the skipped edge, say $c_i=(u_i,u_{i+1})$, together with the two edges which detour around $c_i$.  Define function $\phi(c_i)$ for each edge $c_i=(u_i,u_{i+1}) \in E(C)$.  Let $t_i$ and $t_{i+1}$ be the tree edges adjacent to $u_i$ and $u_{i+1}$, respectively.

\begin{figure}[ht]
	\centering
	\includegraphics[page=4,clip,trim= 190 535 200 65]{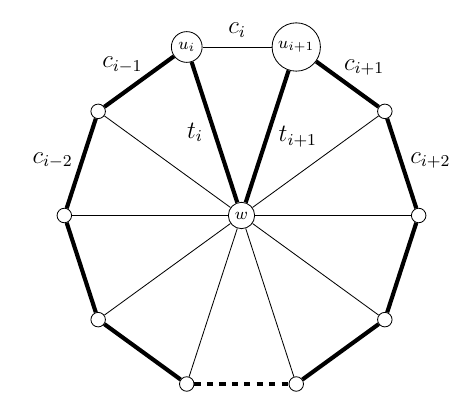}
    	\caption{A tour $\tau$ in a wheel, which skips edge $c_i$.}
    	\label{fig:wheelopt}
\end{figure}

We will define
\begin{eqnarray*}
    \phi(c_i)& = & q(c_{i-2},c_{i-1},t_i) + q(c_{i-1},t_i,t_{i+1}) + q(t_i,t_{i+1},c_{i+1}) \\ &&+ q(t_{i+1},c_{i+1},c_{i+2}) - q(c_{i-2},c_{i-1},c_i) - q(c_{i-1},c_i,c_{i+1})  \\
    && - q(c_i,c_{i+1},c_{i+2}).
\end{eqnarray*}

Then, the optimal tour pair has
\begin{align*}
	z(\tau^*,\rho^*) = q(C) + \min_{i:c_i\in C} \left\{ \phi(c_i) + \min_{\text{feasible }\rho} \sum_{u_j\in C} \beta_{u_j}(\rho_{u_j}) \right\}.
\end{align*}
\
Suppose that we fix an edge $c_i$ in $\tau$.  Then $H$ can be considered to be a fan as shown in Figure~\ref{fig:wheelfan}.

\begin{figure}[ht]
	\centering
	\includegraphics[page=4,clip,trim= 185 355 190 285]{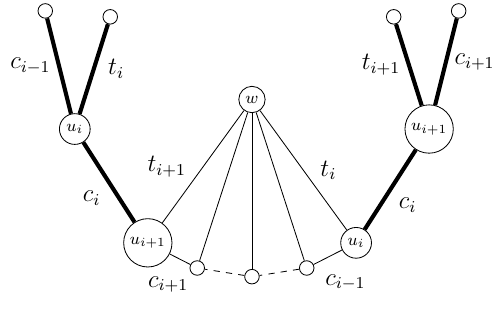}
	\caption{Wheel $H$ with centre $w$ considered as a fan.  Note that edge $c_i$ is fixed and the fan `wraps' around the wheel, reusing edges $c_{i-1},c_i,c_{i+1},t_{i-1},t_i$ and $t_{i+1}$.}
	\label{fig:wheelfan}
\end{figure}

Fix edge $c_r$ in $\tau$ and consider $H$ to be a fan $F_{c_r}$.  Then the minimum tour in $H$ can be determined by calculating the minimum of the minimum centre-traversal of $F_{c_r}$ and the tour which bypasses $c_r$ (using, say, fan $F_{c_1}$).  This can be computed in $O(r)$-time using the pseudofan technique described above.

The preceeding discussion yields the following algorithm.

\begin{algorithm}[H]
\caption{HalinTSP(3)($H$, $q$, $\beta$)}
\begin{algorithmic}
\INPUT Halin graph $H$, quadratic cost function $q$, and penalty function $\beta$
\IF{$H$ is a wheel}
\STATE Use the wheel procedure to find an optimal tour $\tau$ in $H$
\ELSE
\STATE Let $F$ be a fan in $H$
\STATE Contract $F$ to a single node $v_F$, using the Case 1 procedure.  That is, assign penalties $\beta$ to $v_F$, assign costs of $0$ to all triples in $H/F$ and assign costs of $q$ to all remaining triples which are $3$-neighbours in $H$.
\STATE HalinTSP(3)($H/F$, $q$, $\beta$)
\ENDIF
\STATE Expand all pseudonodes in reverse order and update $\tau$
\RETURN $\tau$
\end{algorithmic}
\end{algorithm}

Each time that a fan is contracted, the number of tree nodes is reduced by 1, and hence the fan contraction operation is performed one less than the number of non-leaf nodes in $H$.  The fan contraction operation can be performed in $O(|F|)$-time, and each time it is performed, the number of nodes in $H$ is reduced by $|F|-1$.  Since the wheel procedure takes $O(n)$-time, the total time for the algorithm is $O(n)$.

\subsection{TSP($k$)}
We now show that the previous ideas can be extended to solve TSP($k$).  For any subgraph $G$ of $H$, let $P_k(G)$ be the collection of all distinct candidate $k$-paths in $G$.   For each candidate $k$-path $(e_1,e_2,\ldots\,e_k)\in P_k(H)$, define
\begin{align} \label{eq:tspk}
\begin{split}
q(e_1,e_2,\ldots,e_k) =& q(e_1,e_k) + \frac{q(e_1,e_{k-1}) + q(e_2,e_{k})}{2} + \frac{q(e_1,e_{k-2}) + q(e_2,e_{k-1}) + q(e_3,e_k)}{3}\\&+ \ldots + \frac{c(e_1) + c(e_2) + \ldots + c(e_k)}{k}.
\end{split}
\end{align}

Now consider the simplified problem:
\begin{eqnarray*}
STSP(k): &\textrm{Minimize }   &  \sum_{(e_1,e_2,\ldots,e_k) \in P_k(\tau)}q(e_1,e_2,\ldots,e_k) \\
&\textrm{Subject to } &  \tau \in \mathcal{F}.
\end{eqnarray*}

\begin{theorem}
    Any optimal solution to the STSP(k) is also optimal for TSP(k).
\end{theorem}

\begin{proof}
    Using equation~(\ref{eq:tspk}), the proof of this follows along the same way as that of Theorem~\ref{thm:stsp} and hence is omitted.\qed
\end{proof}
As a result of Corollary~\ref{cor:kpaths}, the preceding algorithm can be extended to solve TSP($k$) by extending the penalty functions at outer nodes to accommodate subpaths of length $2^{\lceil (k+1)/2 \rceil}$.

A complete description of these varies as the information that needs to be stored is more involved and is hence omitted.  Some details however will be available in~\cite{Woods2017}.  The complexity increases by a factor of $(2^{\lceil (k+1)/2 \rceil})$, which is constant for fixed $k$ and polynomially bounded when $k=t\log n$.

\section{TSP(k) on fully reducible graph classes}\label{ReducibleSec}

We say that a class $\mathcal{C}$ of 3-connected graphs is \textit{fully reducible} if it satisfies the following:
\begin{enumerate}
\item If $G\in \mathcal{C}$ has a 3-edge cutset which partitions $G$ into components $S$ and $\bar{S}$, then both $G/S$ and $G/\bar{S}$ are in $\mathcal{C}$ and we call $G$ a \emph{reducible} graph in $\mathcal{C}$; and
\item TSP can be solved in polynomial time for the graphs in $\mathcal{C}$ that do not have non-trivial 3-edge cutsets, i.e.~3-edge cutsets that leave both components Hamiltonian. We call such graphs \emph{irreducible}.
\end{enumerate}

For instance, Halin graphs can be understood as graphs built up from irreducible fans connected to the remainder of the graph via $3$-edge cutsets.  Cornuejols et al.~\cite{Cornuejols:1985} show that the ability to solve TSP in polynomial time on irreducible graphs in $\mathcal{C}$ allows to solve TSP in polynomial time on all of $\mathcal{C}$ using facts about the TSP polyhedron.

We remark that the algorithm of Section $4$ can be used to show a somewhat similar statement for TSP($k$).  Here we consider a graph class $\mathcal{C}$ that is \emph{fully $k$-reducible} in the sense that either it can be subdivided into irreducible graphs via $3$-edge cutsets, or it is irreducible and it is possible to solve the $k$-neighbour Hamiltonian path problem in polynomial time.

This requires solving the following problem:
\begin{eqnarray*}
MTSP(k): &\textrm{Minimize }   &   \sum_{(e_1,\ldots,e_k) \in P_k(\tau)}q(e_1,\ldots,e_k) + \sum_{i\in V}\mathcal{P}_i(\tau) \\
&\textrm{Subject to } &  \tau \in \mathcal{F}
\end{eqnarray*}
where $\mathcal{P}_i(\tau)$ is a penalty function for the pseudonode which depends on how tour $\tau$ traverses $i$, analogous to the construction for the $3$-neighbour TSP of section~\ref{ksec}.

We recursively perform the contraction operation on the irreducible subgraphs of $G$, storing the necessary tour information in the penalty at the resulting pseudonode.  A similar result to Corollary~\ref{cor:kpaths} may be derived to show that for any fixed $k$, this requires a polynomial number of penalties.  The least cost traversals of $S$ can be computed in polynomial time using a generalization of the pseudo-fan strategy above.

Suppose the contraction operation is performed on a subgraph of size $r$ in time $O(P(r))$, where $P(r)$ is a polynomial in $r$.  Each time this operation is performed, the number of nodes in the graph is reduced by $r$.  This operation is performed at most $n$ times and it follows that the entire algorithm can be performed in polynomial time.

\section{Conclusions}
In this paper, we have shown that QTSP is NP-hard even when the costs are restricted to taking 0-1 values on Halin graphs.  We have presented a polynomial time algorithm to solve a restriction of QTSP, denoted TSP($k$) on any fully $k$-reducible graph for any fixed $k$.  To illustrate this, we have given an algorithm which solves TSP(3) on a Halin graph in $O(n)$ time.

The $k$-neighbour bottleneck TSP on a Halin graph can be solved by solving $O(\log(n))$ problems of the type TSP($k$).  However, it is possible to solve the problem faster.  Details will be reported elsewhere.

We would also like to thank Ante \'{C}usti\'{c} for his useful comments.

\bibliographystyle{plain}

\end{document}